\documentclass[preprint]{imsart}

\usepackage[utf8]{inputenc}
\usepackage[colorlinks,citecolor=blue,urlcolor=blue,linkcolor=magenta]{hyperref}
\usepackage[abs]{overpic}
\usepackage[authoryear]{natbib}
\usepackage{multirow}
\usepackage{amssymb,amsmath, amsthm}
\usepackage{graphicx}
\usepackage{framed} 
\usepackage[english]{babel}
\usepackage{booktabs}
\usepackage{color}
\usepackage{epsfig}
\usepackage{paralist}
\usepackage{verbatim}
\usepackage{mathabx}  

\newcommand{\si}{\sigma}

\newcommand{\la}{\lambda}

\renewcommand{\phi}{\varphi}

\newcommand{\PP}{\mathbb{P}}

\newcommand{\RR}{\mathbb{R}}

\renewcommand{\tilde}{\widetilde}

\mathchardef\given="626A

\providecommand{\e}{\mathrm{e}}

\newcommand{\bem}{\begin{bmatrix}}
\newcommand{\enm}{\end{bmatrix}}

\newtheorem{thm}{Theorem}[section]

\theoremstyle{definition}

\newtheorem{rem}[thm]{Remark}
\newtheorem{ex}[thm]{Example}

\newcommand{\dd}{{\,\mathrm d}}

\newcommand{\T}{{\prime}}

\newcommand{\di}{\triangle}

\begin{document}

\begin{frontmatter}

\title{On residual and guided proposals for diffusion bridge simulation}

\runtitle{Residual and guided proposals}

\begin{aug}
\author{Frank van der Meulen\ead[label=e1]{f.h.vandermeulen@tudelft.nl} and
Moritz Schauer, \ead[label=e2]{m.r.schauer@math.leidenuniv.nl}}

\runauthor{Van der Meulen and Schauer}

\affiliation{Delft University of Technology and  Leiden University }

\address{Delft Institute of Applied Mathematics (DIAM) \\
Delft University of Technology\\
Mekelweg 4\\
2628 CD Delft\\
The Netherlands\\
\printead{e1}\\[1em]
Mathematical Institute\\
Leiden University\\
P.O. Box 9512\\
2300 RA Leiden\\
The Netherlands\\
\printead{e2}}

\end{aug}

\begin{abstract}
Recently \cite{Whitaker-BA} considered Bayesian estimation of diffusion driven mixed effects models using data-augmentation. The missing data, diffusion bridges connecting discrete time observations, are drawn using a ``residual bridge construct''. In this paper we compare this construct (which we call residual proposal) with the guided proposals introduced in \cite{Schauer}. It is shown that both approaches are related, but use a different approximation to the intractable stochastic differential equation of the true diffusion bridge. It reveals that the computational complexity of both approaches is similar. Some examples are included to compare the ability of both proposals to capture local nonlinearities in the dynamics of the true bridge. 
\medskip

 \end{abstract}


\end{frontmatter}

\numberwithin{equation}{section}


\section{Introduction}

Simulation of diffusion bridges has received considerable attention over the past two decades. An important application lies in Bayesian estimation of parameters when discrete time observation are obtained from a diffusion process. The main difficulty here is the intractability of the likelihood. If instead of discrete time observations  continuous time trajectories of the diffusion process were to be observed, then the  problem would be simplified due to tractability of the likelihood. This has inspired many authors to adopt a data-augmentation strategy, where the latent diffusion paths in between two discrete time observations are introduced as (infinite-dimensional) latent variables (see for instance \cite{RobertsStramer}, \cite{Golightly1}, \cite{MR2422763} and \cite{vdm-s-estpaper}). A non-exhaustive list of references on diffusion bridge simulation is  \cite{Eraker}, \cite{ElerianChibShephard}), \cite{DurhamGallant}, \cite{Clark}, \cite{Bladt2}, \cite{beskos-mcmc-methods}, \cite{HaiStuaVo09}, \cite{Schoenmakers}, \cite{LinChenMykland}), \cite{BeskosPapaspiliopoulosRobertsFearnhead}, \cite{DelyonHu}, \cite{lindstrom}, \cite{Schauer} and \cite{Whitaker}. A somewhat more detailed overview is given in the introductory section of \cite{Whitaker}. 

As direct and exact simulation of diffusion bridges in complete generality is infeasible, most approaches consist of simulating a process that resembles the true diffusion bridge (a proxy) and correcting for the discrepancy by either weighting, accept/reject methods or the Metropolis-Hastings algorithm. 
In the following, any stochastic process which is used as a proxy for the diffusion bridge is called a proposal. 

Here we are particularly interested in the connections between the ``guided proposals''  and ``residual proposals'' introduced in \cite{Schauer} and \cite{Whitaker} respectively (the terminology ``residual  proposals'' being ours). This is motivated by the recent paper of \cite{Whitaker-BA} where the authors apply residual proposals  for Bayesian estimation of diffusion driven mixed effects models.

They note that finding a guided proposal that is both accurate and computationally efficient may be difficult in practice (at the bottom of page 441 of their paper.) We show in fact that  a variation of the residual proposals used in \cite{Whitaker-BA} can be obtained as special cases of guided proposals. Moreover, we explain that  the choices to be made for  constructing either residual or guided proposals are similar, as is their computational cost. For that reason we don't see a strong argument to dismiss guided proposals. On the contrary, we present some examples where guided proposals appear to perform better than  residual proposals. That is not to say that residual proposals cannot give a substantial improvement to the modified diffusion bridge of \cite{DurhamGallant} which ignores nonlinearities in the drift. 

The structure of this paper is as follows. In section \ref{sec:general} we introduce notation and concepts used throughout.  In sections \ref{sec:resbridge} and \ref{sec:guided} we recap the approaches introduced in \cite{Whitaker} and \cite{Schauer} respectively. Moreover, in section \ref{sec:resbridge} we derive the likelihood of residual proposals with respect to the true diffusion bridge without resorting to discretisation. For guided proposals we give some new insights for an efficient implementation in subsection \ref{sec:guided-impl}. 
 In section \ref{sec:connections} we expose the relations and similarities of these two approaches and illustrate these with a couple of numerical examples. 

\section{Diffusion bridges}\label{sec:general}
In the following we assume the dynamics of the diffusion process $X=(X_t,\, t\ge 0)$ are governed by the stochastic differential equation (sde)
\[ \dd X_t = b(t,X_t) \dd t + \si(t,X_t) \dd W_t, \qquad X_0=u. \]
Here $b\colon [0,\infty) \times \RR^d \to \RR^d$ and $\si\colon[0,\infty)\times \RR^d \to \RR^{d\times d'}$ are referred to as the drift and dispersion coefficient respectively and $W$ is a $d'$-dimensional vector of uncorrelated Wiener processes.  If we condition the process on hitting the point $v$ at time $T$, the conditioned process is denoted by $X^\star$. 

 The process $X^\star$ is the solution of the SDE
\begin{equation}\label{eq:xstar} \dd X^\star_t = b(t,X^\star_t) \dd t + a(t,X^\star_t) r(t,X^\star_t) \dd t+ \si(t,X^\star_t) \dd W_t, \qquad X^\star_0=u, \end{equation} where $a(t,x) = \si(t,x) \si(t,x)^\T$ is diffusion coefficient,  $r(t,x)=\nabla_x \log p(t,x; T,v)$ and $p$ denotes the transition densities of the diffusion $X$. Hence, the sde for the bridge process $X^\star$ has an additional pulling term in the drift to ensure it hits $v$ at time $T$. Note that as  $p$ is intractable, it is impossible to obtain diffusion bridges by discretising the sde for $X^\star$. 

However, in case $b\equiv 0$ and $\si$ is constant (so that $X$ is a scaled Wiener process), transition densities are tractable and the diffusion bridge has dynamics
\[ \dd X^\star_t = \frac{v-X^\star_t}{T-t} \dd t + \si \dd W_t, \qquad X^\star_0=u. \]
This has motivated \cite{DelyonHu} to consider proposals $\bar{X}$ driven by the sde
\[ \dd \bar{X}_t = \la b(t,\bar{X}_t) \dd t +   \frac{v-\bar{X}_t}{T-t} \dd t + \si(t,\bar{X}_t) \dd W_t,\qquad  \bar{X}_0=u, \]
where either $\la=0$ or $\la=1$ is chosen. Realisations of $\bar{X}$ can then be used in a Metropolis-Hasings acceptance ratio, or by importance sampling as a proxy for $X^\star$. Key to the feasibility of this approach is the derivation of the likelihood ratio $\frac{\dd \PP^\star}{\dd \bar{\PP}^{\phantom \star}}(\bar{X})$, where $\PP^\star$ and $\bar{\PP}$ denote the laws of $X^\star$ and $\bar{X}$ on $C[0,T]$ respectively and $\bar{X}=(\bar{X}_t,\, t\in [0,T])$. For further details we refer to \cite{PapaRobertsStramer} and \cite{vdm-s-estpaper}. 

The choice $\la=0$ appears to be most popular in the literature, especially when applied with a particular correction to Euler discretisation proposed by \cite{DurhamGallant} yielding the ``modified diffusion bridge''. Clearly, this approach completely ignores the drift and can only be efficient when the drift is approximately constant on $[0,T]$. Its simplicity is however attractive as there is no tuning parameter. 


\section{Residual proposals} \label{sec:resbridge}
The key idea in \cite{Whitaker} is to define a proposal $X^\di$ that does take the drift of the diffusion into account by using the decomposition
\[ X^\di_t = x(t) + C_t, \]
where $x(t)$ is defined as the solution to
\begin{equation}
	\label{eq:dyn-system}
 \dd x(t) = b(t,x(t)) \dd t, \qquad x(0)=u
 \end{equation}
and the residual process $C$ is given by
\begin{equation}\label{eq:residual-proc}
	d C_t = \frac{v-x(T)-C_t}{T-t} \dd t + \si(t,x(t) +C_t) \dd W_t,\qquad C_0=0. 
\end{equation} 
In this approach, first the process $x(t)$ is determined, either in closed form or by a numerical procedure. Subsequently, the sde for $C$ is solved using Euler discretisation. 
It is easily seen that $X^\di$ satisfies the sde
\begin{equation}\label{eq:resprop}
	d X^\di_t= b(t,x(t)) \dd t + \frac{v-X^\di_t - (x(T)-x(t))}{T-t} \dd t + \si(t,X^\di_t) \dd W_t, \qquad X^\di_0=u. 
\end{equation} 
Compared to the sde for $X^\star$ given in \eqref{eq:xstar}, it appears that
\begin{enumerate}
  \item  $b(t,X^\star_t)$ is approximated by $b(t,x(t))$,
  \item the pulling term is replaced with 
 \[ \frac{v-X^\di}{T-t} - \frac1{T-t} \int_t^T b(s,x(s)) \dd s\]
 where we recognise the first term as the  pulling term appearing in the proposals introduced by \cite{DelyonHu}.
\end{enumerate}

Define $\kappa(s,x) = (T-s)^{-1}(v-x)$ (considering $T$ and $v$ to be fixed). Denote by $\varphi(x; \mu,a)$  the value of the normal density with mean vector $\mu$ and covariance matrix $a$, evaluated at $x$. 
Denote the law of $X^\di$ on $C[0,T]$ by $\PP^\di$.  The following theorem reveals that $\PP^\star$ is absolutely continuous with respect of $\PP^\di$.  This is crucial, as otherwise of course the process $X^\di$ cannot be used as a proposal. 
\begin{thm}\label{thm:lr-resproc}
Assume the $C^{1,2}$-function $\si$ with values in $\RR^{d\times d}$ is bounded and has bounded derivatives and  is invertible with a bounded inverse. Assume the function $b$ is locally
Lipschitz with respect to $x$ and is locally bounded. Finally, assume the sde for $X$ admits a strong
solution.  Then
\[ \frac{\dd \PP^\star}{\dd \PP^\di}(X^\di)= \frac{\phi(v; u, a(0,u)T)}{p(0,u;T,v)} \sqrt{\frac{|a(0,u)|}{|a(T,v)|}}  \Psi_1(X^\di) \Psi_2(X^\di),\]
with
\begin{align*}  \Psi_1(X^\di)&= \int_0^t b(s,X^\di_s)^\T a^{-1}(s,X^\di_s) \dd X^\di_s- \frac12 \int_0^t b(s,X^\di_s)^\T  a^{-1}(s,X^\di_s) b(s,X^\di_s)\dd s  \\ & -\frac12 \int_0^t \kappa(s,X^\di_s)^\T  \diamond \dd a^{-1}(s,X^\di_s) (v-X^\di_s),  \end{align*}
where  the $\diamond$-integral is obtained as the limit of sums where the integrand is computed at the right limit of each time interval as opposed to the left limit used in the definition of the It\=o integral;
and 
\begin{align*} &	\Psi_2(X^\di) =\exp\left(\frac12 \int_0^T f(s,X^\di_s)^\T a(s,X^\di_s)^{-1} f(s,X^\di_s) \dd s \right.\\ & \left. \qquad -\int_0^T f(s,X^\di_s)^\T a(s,X^\di_s)^{-1} \dd X^\di_s + \int_0^T f(s,X^\di_s)^\T a^{-1}(s,X^\di_s) \kappa(s,X^\di_s) \dd s\right).\end{align*}
Here, $f(s,x) =b(s,x(s)) - (T-s)^{-1}(x_T-x)$ with $x(s)$ as defined in \eqref{eq:dyn-system}.  \end{thm}
\begin{proof}
Absolute continuity of $\PP^\star$ with respect to $\bar\PP$ was proved in \cite{DelyonHu} under the stated assumptions (which is assumption 4.2 in that paper). Here we consider the proposal with $\la=0$ and the likelihood ratio being proportional to $\Psi_1(X^\di)$. The normalising constant in the likelihood ratio was derived in \cite{PapaRobertsStramer}. 
Now $\bar{\PP}$ is absolutely continuous with respect to $\PP^\di$ with Radon-Nikodym derivative
$\left(\mathrm{d} \bar{\PP} /\mathrm{d}\PP^\di\right)(X^\di) = \Psi_2(X^\di)$. 
The result now follows from
\[ 		\frac{\mathrm{d}\PP^\star}{\mathrm{d}\PP^\di}(X^\di)=  \frac{\mathrm{d}\PP^\star}{\mathrm{d}\bar{\PP}^{\phantom\star}}\frac{\mathrm{d}\bar{\PP}^{\phantom\di}}{\mathrm{d}\PP^\di}(X^\di). \]
\end{proof}

The preceding theorem 
The unknown transition density only shows up as a multiplicative constant in the denominator and henceforth cancels in all calculations for performing MCMC for Bayesian estimation of diffusion processes with a data-augmentation strategy 
(Cf.\ \cite{vdm-s-estpaper}). 

\subsection{Improved residual proposals using the LNA}

Clearly, any deterministic trajectory $(x(t),\, t\in [0,T])$ can be used and the definition of residual proposals is not restricted to \eqref{eq:dyn-system}.  One particular choice proposed by \cite{Whitaker} is based on the linear noise approximation (abbreviated by LNA, see for instance \cite{VanKampen}). It uses the decomposition
\[ X^\di_t = x(t) + \rho(t)+ C_t, \]
where $x(t)$ is defined in \eqref{eq:dyn-system} and   $\rho(t)= E [Y_t \mid Y_T=v]$ with 
\[ \dd Y_t = V(t,x(t)) Y_t  \dd t + \si(x(t)) \dd W_t, \qquad Y_0=u. \]
Here $V(t,y)$ is the matrix with elements $V(t,y)_{i,j}=\partial b_i(t,y)\,/\, \partial y_j$ for $y \in \RR^d$. 


\section{Guided proposals} \label{sec:guided}
The basic idea in \cite{Schauer} is to replace the generally intractable transition density $p$ that appears in \eqref{eq:xstar} by  the transition density of an auxiliary diffusion process $\tilde{X}$  with tractable transition densities. Assume $\tilde{X}$ satisfies the sde
\[\dd \tilde{X}_t = \tilde{b}(t, \tilde{X}_t) \dd t + \tilde{\si}(t, \tilde{X}_t) \dd W_t\] and  denote the transition densities of $\tilde{X}$ by $\tilde{p}$. Define the process $X^\circ$ as the solution of the sde 
\begin{equation}\label{eq:sde-xcirc} \dd X^\circ_t = b^\circ(t,X^\circ_t) \dd t + \si(t,X^\circ_t) \dd W_t,\qquad X^\circ_0=u,\end{equation}
with 
\[ 		b^\circ(t, x) = b(t,x) + a(t, x) \nabla_x  \log  \tilde p(t,x;T,v).\]
A process $X^\circ$ constructed in this way is referred to as a {\it guided proposal}. 

Let $\tilde{a}=\tilde{\si}\tilde{\si}^\T$. 
It is proved in \cite{Schauer} that if $\tilde{a}(T)= a(T,v)$ (and a couple of other somewhat technical conditions) 
\begin{equation}\label{likeli}
\frac{\rm{d} \PP^\star}{\rm{d} \PP^\circ}(X^\circ) = \frac{\tilde{p}(0, u; T,v)}{p(0, u; T,v)}
 \exp\left( \int_0^T G(s,X^\circ_s) \dd s\right),
\end{equation}
where $G$ depends on $b$, $a$,  $\tilde{b}$ and $\tilde{a}$, but not on $p$. Note that the unknown transition density only appears as a multiplicative constant in the dominator.

Regarding the choice of $\tilde{X}$, the class of linear processes,
\begin{align}\label{linsdehomog}
		\dd  \tilde X_t = \tilde B(t) \tilde X_t\dd t  + \tilde\beta(t)  \dd t +  \tilde \sigma(t)  \dd  W_t,
\end{align}
 is a flexible class  with known transition densities and its induced guided proposals  will be referred to as {\it linear guided proposals}. Not surprisingly, the efficiency of guided proposals depends on the choice of $\tilde{B}$, $\tilde\beta$ and $\tilde\si$. 
 A particularly simple type of linear guided proposals is obtained upon choosing $d\tilde{X}_t = \tilde\beta(t) \dd t+ \si(T,v) \dd W_t$. For this  choice
\begin{equation}\label{eq:simple-linear} b^\circ(t,x)= b(t,x)  + a(t, x) a(T,v)^{-1}\frac{v-x- \int_t^T \tilde\beta(s) \dd s}{T-t}. \end{equation}

An alternative is to use the LNA, where 
  $b(t,X_t)$ is approximated by $b(t,x(t)) +V(t,x(t)) (X_t-x(t))$. Here,  $V(t,y)$ is the matrix with elements $V(t,y)_{i,j}=\partial b_i(t,y)\,/\, \partial y_j$ for $y \in \RR^d$.  This gives linear guided proposals with 
\begin{equation}
	\label{eq:gp-lna} 
	\tilde\beta(t)= b(t,x(t)) - V(t,x(t)) x(t) \quad \text{and} \quad \tilde B(t)=V(t,x(t)).
\end{equation} 
An easy choice for $\tilde\si$ is $\tilde\sigma(t)=\sigma(T,v)$. 

\begin{rem}
\cite{Whitaker}  compare various proposals, among which the guided proposal with $\tilde\beta$ and $\tilde{B}$ as in \eqref{eq:gp-lna} and $\tilde\sigma(t)=\sigma(t,x(t))$. However, as correctly remarked in their paper, such proposals do not satisfy the key requirement $\tilde{a}(T,v)=a(T,v)$.  Therefore, the measure $\PP^\star$ will be singular with respect to the law of such proposals. As a fix to this they also considered proposals where  the deterministic process is continually restarted and simulated on $[t,T]$, where $0<t<T$. This comes naturally with a huge increase in computing time. We would rather propose to either take $\tilde\sigma(t)=\sigma(T,v)$ or 
\[ \tilde\sigma(t) = \begin{cases} \si(t,x(t)) & \qquad \text{for}\quad t\in [0,t_0] \\
 	\frac{\si(T,v)-\si_0}{T-t_0} (t-t_0) +\si_0& \qquad \text{for}\quad t\in (t_0,T] 	
 \end{cases},
\] where $\si_0=\si(t_0,x(t_0))$. The latter choice naturally raises the question on how to choose $t_0$ but for sure satisfies the requirement for absolute continuity of $\PP^\star$ with respect to the law of the proposal. 

\end{rem}

\subsection{Implementation}\label{sec:guided-impl}
For guided proposals, the computational cost consists of discretising the sde \eqref{eq:sde-xcirc} and evaluating the likelihood ratio in \eqref{likeli}.  For linear guided proposals, there is a convenient expression for $\tilde{r}$, i.e.\ $\tilde{r}(s,x)=\tilde{H}(s) (v(s)-x)$, where 
 $\tilde{H}(t)=K^{-1}(t)$ with
\begin{align*}
K(t) &= \int_t^T \Phi(t,s) \tilde{a}(s) \Phi(t,s)^\T \dd s
\\
v(t) &= \Phi(t,T)v - \int_t^T \Phi(t,s)\tilde\beta(s) \dd s
\end{align*}
and $\dd \Phi(t) = \tilde{B}(t) \Phi(t) \dd t$, $\Phi(0)=I$. 

We propose computing $K$ and $v$ recursively backwards using the  differential equations
\begin{equation}\label{eq:ode-K}
\frac{\dd}{\dd t} K(t) = \tilde{B}(t) K(t) +  K(t) \tilde{B}(t)^\T -  \tilde{a}(t), \qquad K(T) = 0
\end{equation} and 
\[
\frac{\dd}{\dd t} v(t) = \tilde{B}(t) v(t)  + \tilde{\beta}(t), \qquad v(T) = v.
\]
These equations can be discretised using for example an explicit backwards Runge-Kutta scheme. Next, $\tilde{r}(s,x)$ can be obtained from solving $K(s) \tilde{r}(s,x)=v(s)-x$ for which  the Cholesky decomposition can be used. 

In case $\tilde{B}$ and $\tilde{a}$ are not time-dependent (that is, constant), then $K$ (and of course $v$) can also be computed in closed form. Let $\Lambda$ be the solution to the continuous Lyapunov equation
\[
\tilde B\Lambda + \Lambda \tilde B' + \tilde a = 0.
\]
Then, as verified by direct computation
\[
K(t) = 
     \e^{-(T-t)\tilde B}\Lambda\e^{-(T-t)\tilde B'}- \Lambda
\]
and, with $\mu$ solving $\tilde B\mu + \beta = 0$, 
\[
v(t) =  \e^{-(T-t)\tilde B}(v - \mu) + \mu,
\]
LAPACK includes the function \texttt{trsyl!} to solve the continuous Lyapunov  equation in an efficient manner, and correspondingly  many high level computing environments provide this functionality. Nevertheless, computing $K$ from the closed form expression is  computationally more demanding than discretising \eqref{eq:ode-K}, at the cost of allowing for discretisation error. However,    when using  a 3rd order Runge-Kutta scheme (say), this error is small compared to the discretisation error induced by the  Euler-Maruyama scheme for $X^\circ$.


\section{Connections} \label{sec:connections}

The idea of using  $x(t)$, the solution of the  dynamical system given in \eqref{eq:dyn-system}, can also be used with linear guided proposals. Indeed, in section  4.4 of \cite{vdm-s-estpaper} it was proposed to take guided proposals with 
\begin{equation}\label{eq:choice-linprocess}
	\tilde\beta(t)=b(x(t)), \qquad  \tilde{B}\equiv 0 \qquad \text{and} \qquad  \tilde\sigma = \sigma(T,v).
\end{equation} 
It  follows immediately from \eqref{eq:simple-linear} that the resulting guided proposal satisfies the sde
\begin{multline*}
	\dd X^\circ_t= b(t,X^\circ_t) \dd t  \\ + a(t, X^\circ_t) a(T,v)^{-1}\frac{v-X^\circ_t- (x(T)-x(t))}{T-t}\dd t + \si(t,X^\circ_t) \dd W_t,\qquad X^\circ_0=u.
\end{multline*}  
This is to be compared to the sde for the residual proposal given in \eqref{eq:resprop}. From this we infer the following:
\begin{itemize}
  \item The computational effort for computing either guided or residual proposals is of the same order. The inverse of $a(T,v)$ needs to be computed only once and besides that the only extra calculation required by guided proposals is premultiplication by $a(t,X^\circ_t)$. 
  \item If $\si$ is constant and the residual process defined in \eqref{eq:residual-proc} is redefined to satisfy the sde
\begin{multline*}
	\dd C_t = \left(b(t,x(t)+C_t)-b(t,x(t))\right) \dd t \\+\frac{v-x(T)-C_t}{T-t} \dd t + \si(t,x(t) +C_t) \dd W_t,\qquad C_0=0,
\end{multline*}
then these adjusted residual proposals  are in fact linear guided proposals with $\tilde\beta$, $\tilde{B}$ and $\tilde\si$ specified in \eqref{eq:choice-linprocess}. Note that the difference with the definition of residual proposals is the addition of the  first term on the right-hand-side in the preceding display. 
\item If the residual process defined in \eqref{eq:residual-proc} is redefined to satisfy the sde
\begin{multline*}
	\dd C_t = \left(b(t,x(t)+C_t)-b(t,x(t))\right) \dd t \\+a(t,x(t)+C_t) a(T,v)^{-1}\frac{v-x(T)-C_t}{T-t} \dd t \\+ \si(t,x(t) +C_t) \dd W_t,\qquad C_0=0, 
\end{multline*}
then then these adjusted residual proposals  are simply  linear guided proposals with $\tilde\beta$, $\tilde{B}$ and $\tilde\si$ specified in \eqref{eq:choice-linprocess} (also when $\si$ is nonconstant).

\end{itemize}

\bigskip

We conclude with a couple of examples in which we investigate the behaviour of guided and residual proposals. 

\begin{ex}\label{ex:ouExample}
Suppose $\si$ is constant and  $b(t,x)=-\alpha x$. 
Suppose we condition the process to hit $1$ at time $T$, when $X_0=0$. The dynamical system's solution is given by $x(t)=u e^{-\alpha t}$ and the sde for the (true) bridge can be computed in closed form
\[ \dd X^\star_t = -\alpha X^\star_t \dd t +2\alpha \frac{e^{-\alpha(t-T)} v -X^\star_t}{e^{-2\alpha(t-T)}-1} \dd t + \si \dd W_t, \qquad X_0=u.\]
We took $u=0.1$, $v=1$, $T=3$, $\si=0.1$, $\alpha=2$  and simulated a realisation of both the  guided and residual proposal using the same Wiener increments (Euler discretisation with time step $10^{-4}$ was used). We then repeated this $4$ more times and also simulated $5$ realisations of true bridges; the result is in figure \ref{fig:ouExample}. Clearly, both proposals deviate from true bridges, though the guided proposals resemble true bridges  better  than residual proposals.  
\begin{figure}
\begin{center}
\includegraphics[scale=0.8]{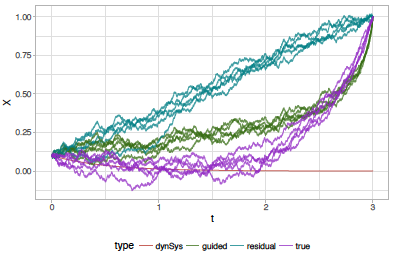}		
\end{center}
\caption{Simulation results for $b(x)=-2 x$ and $\sigma=0.1$. ``dynSys'' refers to the solution of the dynamical system. For true bridges, guided proposals and residual proposals five simulated realisations are shown.}
\label{fig:ouExample}
\end{figure}

\end{ex}

\begin{ex}
Assume again that $\si$ is constant, $X_0=0$ and  $b(t,x)=-\sin(2\pi x)$. In this case $x(t)\equiv 0$ and then residual proposals have dynamics
\[ \dd X^\di_t =\frac{v-X^\di_t}{T-t} \dd t + \si \dd W_t, \qquad X^\di_0=0 \]
whereas guided proposals with \eqref{eq:choice-linprocess} take the form
\[ \dd X^\circ_t = b(t,X^\circ_t) \dd t + \frac{v-X^\circ_t}{T-t} \dd t + \si \dd W_t, \qquad X^\circ_0=0. \]
Both of these proposals have been proposed earlier by \cite{DelyonHu}. As the drift has multiple wells and the process starts in one of those, residual proposals will completely miss the nonlinear dynamics of the true bridge. The guided proposals in this example may do slightly better in resembling true bridges, but will in most cases also perform unsatisfactory. Nevertheless, the class of guided proposals is way more flexible and not restricted to the choice in \eqref{eq:choice-linprocess}. For a very similar drift as considered in this example, this was illustrated in section 1.3 of \cite{Schauer}.

\end{ex}

\begin{ex}
Suppose $b(x)=-\frac12 x -\sin(2\pi x)$, $\sigma=0.15$, $u=5$, $v=2$ and $T=5$. In this case the solution of the (deterministic) dynamical system is not available in closed form. However, it can easily be obtained using the $4$-th order Runge-Kutta scheme. In the top left panel of figure \ref{fig:ouSine} the obtained solution is displayed. Note that  $b(x)=0$ for $x \approx 1.9$ which is close to $v$. Hence bridges follow approximately the solution of the dynamical system, which makes proposals based on this solution suitable.  
In this case we can also simulate true bridges using the method by \cite{Bladt} which is implemented in the {\tt sde}-library in {\tt R}. In figure \ref{fig:ouSine},  $25$ realisations of the bridges and both proposals are shown (Euler discretisation with time step $10^{-4}$ was used).  In this case it is guided proposals appear to perform somewhat better than residual proposals.  

\end{ex}
\begin{figure}
\begin{center}
\includegraphics[scale=0.8]{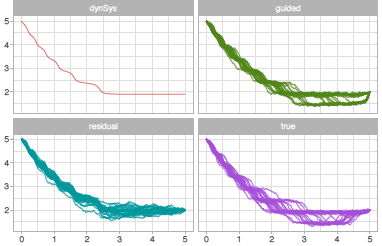}		
\end{center}
\caption{Simulation results for $b(x)=-\frac12 x -\sin(2\pi x)$ and $\sigma=0.15$. Top left: solution of the dynamical system.  Top right: $25$ realisations of guided proposals. Bottom left: $25$ realisations of residual proposals. Bottom right: $25$ realisations of diffusion bridges using the method by \cite{Bladt}. }
\label{fig:ouSine}
\end{figure}

\section*{Acknowledgement} 

This work was partly supported by the Netherlands Organisation for Scientific Research (NWO) under the research programme ``Foundations of nonparametric Bayes procedures'', 639.033.110 and by the ERC Advanced Grant ``Bayesian Statistics in Infinite Dimensions'', 320637.

\bibliographystyle{harry}

\bibliography{lit}

\end{document}